\documentclass[journal,twocolumn,10pt]{IEEEtran}

\usepackage{epsfig}
\usepackage{amsmath}
\usepackage{amssymb}
\usepackage{amsthm}
\usepackage{amsfonts}
\usepackage{mathrsfs}
\usepackage{chgbar}
\newtheorem{lem}{Lemma}[section]

\usepackage{picinpar}
 \IEEEoverridecommandlockouts
\input epsf

\begin{document}
\title{Generalised Precoded Spatial Modulation for Integrated Wireless Information and Power Transfer}

\author{Rong Zhang, Lie-Liang Yang and Lajos Hanzo \\
Communications, Signal Processing and Control, School of ECS, University of Southampton, SO17 1BJ, UK \\
Email: {rz,lly,lh}@ecs.soton.ac.uk, http://www-mobile.ecs.soton.ac.uk\protect
\thanks{{The financial support of the EPSRC under the India-UK Advanced Technology Centre (IU-ATC), that of the EU under the Concerto project as well as that of the European Research Council's (ERC) Advanced Fellow Grant is gratefully acknowledged.}}
}

\markboth{}{Generalised Precoded Spatial Modulation for Integrated Wireless Information and Power Transfer}

\maketitle

\begin{abstract}
Conventional wireless information transfer by modulating the amplitude, phase or frequency leads to an inevitable Rate-Energy (RE) trade-off in the presence of simultaneous Wireless Power Transfer (WPT). In echoing Varshney's seminal concept of \textit{jointly} transmitting both information and energy, we propose the so-called Generalised Precoded Spatial Modulation (GPSM) aided Integrated Wireless Information and Power Transfer (IWIPT) concept employing a power-split receiver. The principle of GPSM is that a particular subset of \textit{Receive} Antennas (RA) is activated and the activation pattern itself conveys useful information. Hence, the novelty of our GPSM aided IWIPT concept is that RA pattern-based information transfer is used in addition to the conventional waveform-based information carried by the classic $M$-ary PSK/QAM modulation. Following the Radio Frequency (RF) to Direct Current (DC) power conversion invoked for WPT at the power-split receiver, the non-coherent detector simply compares the remaining received power accumulated by each legitimate RA pattern for the sake of identifying the most likely RA. This operation is then followed by down-conversion and conventional Base Band (BB) $M$-ary PSK/QAM detection. Both our analysis and simulations show that the RA pattern-based information transfer represented in the Spatial Domain (SD) exhibits a beneficial immunity to any potential power-conversion induced performance degradation and hence improves the overall RE trade-off when additionally the waveform-based information transfer is also taken into account. Moreover, we investigate the impact of realistic imperfect Channel State Information at the Transmitter (CSIT) as well as that of the antenna correlations encountered. Finally, the system's asymptotic performance is characterised in the context of large-scale Multiple Input Multiple Output (MIMO) systems.
\end{abstract}

\section{Introduction}
\subsubsection{Motivation}
In thermal and statistical physics, the earliest and most famous thought experiment regarding information and energy was conceived by Maxwell in 1867, which is referred to as "Maxwell's Demon"~\cite{RevModPhys.81.1}, where the second law of thermo-dynamics was hypothetically violated by the bold hypothesis of information to energy conversion. This stimulated further intriguing research in the mid-20th century as to whether information processing itself dissipates energy, which subsequently led to "Landauer's principle"~\cite{5392446,335779a0} suggesting that thermo-dynamically reversible manipulations of information in computing, measurement and communications do not necessarily dissipate energy, since no energy is required to perform mathematical calculations. This principle has thus led to further research in the area of reversible computing and reversible electronic circuits~\cite{a21-saeedi}. Although elusive, a so-called "information heat engine" was experimentally demonstrated in~\cite{nphys1821} for showcasing information to energy conversion with the aid of feedback control. 

Despite the fact that information to energy conversion is in its infancy, it was suggested from a fundamental perspective~\cite{Engineering} that the currently separate treatment of information and energy has to be challenged in the practical design of engineering systems. Naturally, information is carried by attaching itself to a physical medium, such as waves or particles. In molecular and nano communications, information is delivered by conveying encoded particles from the source to destination~\cite{6208883,p84-akyildiz}. Similarly, in optical communications, information is delivered by photons having information-dependent intensities, which may be detected by a photon-counting process~\cite{oe-21-13-15926,oe-21-22-25954}. Given the nature of the process, the system is capable of providing a heating/illumination/propulsion function. Both of the above examples suggest that the underlying matter that carries information can be effectively reused for diverse applications. The explicit concept of transporting both information and energy simultaneously was raised by the authors of~\cite{4595260,6283481,5513714} in the context of power-line communications, where information is conveyed by wires that carry electricity. This topic was recently further extended to wireless communications, since it is desirable that a mobile device is free from being tethered in any way, as prophesied a century ago~\cite{Tesla1908}.  

\subsubsection{Background}
Before delving into the topic of Simultaneous Wireless Information and Power Transfer (SWIPT), a brief historical portrayal of Wireless Power Transfer (WPT) is warranted~\cite{6494253}. The earliest experiments on WPT were conducted by Tesla with the ultimate goal of creating a worldwide wireless power distribution system~\cite{Tesla1914}. In the mid-20th century, one of the most influential projects was the microwave-powered helicopter~\cite{1132833}. Similarly, the National Aeronautics and Space Administration's (NASA) feasibility study of the space-solar program~\cite{1145675} attracted global attention. The need for small to moderate power requirements (upto hundreds of watts) and near-field (upto a few meters) WPT increased substantially owing to the development of electronic devices in the late 20th century. Hence the wireless power consortium was established for promoting the standardization for wireless charging, known as Qi~\footnote{see websites: http://www.wirelesspowerconsortium.com/}. WPT can be carried out in two basic ways, namely based on either electromagnetic \textit{induction} in the form of inductive coupling and resonant coupling relying on coils or with the aid of electromagnetic \textit{radiation} using microwave frequencies by relying on so-called rectennas, which will be discussed in Section~\ref{sec_iwipt_rx}. At the time of writing there are already numerous WPT applications, ranging from wirelessly charged mobile phones~\cite{6529362} to wirelessly powered medical implants~\cite{6178849} and to 'immortal' sensor networks~\cite{6590061} in addition to passive Radio-frequency IDentification (RFID) devices~\cite{4342873}. Furthermore, along with successful start-up companies like WiTricity and Powercast~\footnote{see websites: http://www.witricity.com/ and http://www.powercastco.com/}, it is thus the maturing WPT and wireless communications fields that make the SWIPT an interesting emerging research topic.

However, the current research of SWIPT is still in its infancy and hence there is a paucity of literature, which may be classified as follows. The first set of contributions considered \textit{one-way} SWIPT, where the transmitter simultaneously conveys both information and power to the receiver that may be operated in two different modes~\cite{6489506}, namely either on a time-division basis or on a power-split basis. To be more specific, in the time-division mode, the receiver alternatively and opportunistically acts as an information detector and power converter~\cite{6373669}. On the other hand, in the power-split mode, a portion of the received power is used for powering the receiver and the remaining received power is used for retrieving information~\cite{6567869,1205.0618v3}. The second set of contributions considered \textit{two-way} SWIPT. Specifically, the authors of~\cite{1211.6868v3,1304.7886v2} investigated the scenario, where the transmitter conveys power to the receiver, which is converted to Direct Current (DC) power and reused for the destination's information transmission in the reverse direction. This mode of operation is similar to that of the passive RFID devices and hence it may be referred to as being \textit{two-way half-duplex}. Distinctively,~\cite{6480924} investigated a \textit{two-way full-duplex} operational mode, where a pair of nodes interactively communicates and exchanges power. The final set of treatises considered a range of energy-transfer-aided systems, such as multi-carrier systems~\cite{1303.4006}, relay-assisted systems~\cite{6313575}, interference-limited arrangements~\cite{6571308}, beamforming-assisted regimes~\cite{6211384,1307.6285v1}, multiple access~\cite{6307793} and unicast/multicast scenarios~\cite{6512098}.

\subsubsection{Novelty}
A close inspection of the existing literature reveals that most of the current designs are centered on the information transfer strategy in the presence of power transfer, while satisfying a specific received energy constraint. It would be however more beneficial to aim for the full "integration" of information and power transfer as the ultimate objective in the spirit of Varshney's concept~\cite{6283481}, who proposed that energy and information transfer should be innately inter-linked.  

To be more explicit, in one-way SWIPT equipped with a power-split receiver~\footnote{The time-division receiver is conceptually simple but it has been shown to be sub-optimum~\cite{6283481}. Hence it is not considered in this paper.}, an inevitable Rate-Energy (RE) trade-off has to be obeyed, which simply suggests that the more power is transferred for powering the receiver, the lower the communications rate becomes. This observation is intuitive, since from an information theoretic point of view, less power is left for information transfer, which inevitably leads to a reduced information rate. Naturally, having a reduced power for information transfer implies that the information becomes more prone to errors. One of the fundamental reasons for the RE trade-off is that information is conventionally transferred by modulating the amplitude, phase or frequency, which are naturally subjected to channel-induced distortion, noise and power-splitting. In order to eliminate these impediments, we propose a completely different modulation type for Integrated Wireless Information and Power Transfer (IWIPT), where information is carried not by \textit{waveforms} but by \textit{patterns}.

Explicitly, a multi-antenna aided system is conceived for IWIPT, where the specific Receive Antenna (RA) patterns are used for conveying information, while simultaneously transferring power.
\begin{itemize}  
\item We achieve this with the aid of our recently proposed Generalised Pre-coding aided Spatial Modulation (GPSM) scheme~\cite{6644231}, where information is conveyed by appropriately selecting the RA indices for information transfer in addition to the information carried by the conventional $M$-ary PSK/QAM modulation~\footnote{Spatial Modulation~\cite{4382913,Renzo} conveys extra information by appropriately selecting the \textit{transmit} antenna indices, but our GPSM scheme conveys extra information by appropriately selecting the \textit{receive} antenna indices~\cite{6644231}.}.

\item We show that the RA pattern-based information transfer represented in the Spatial Domain (SD) exhibits a beneficial immunity to any performance degradation imposed by the Radio Frequency (RF) to DC power conversion in a power-split receiver and that our GPSM scheme is capable of improving the overall RE trade-off, when additionally the waveform-based information transfer represented in the classic time domain is also taken into account.
\end{itemize}

\subsubsection{Organisation} The rest of our paper is organised as follows. In Section \ref{sec_systemmodel}, we introduce the underlying concept as well as the transceiver architecture of the GPSM aided IWIPT system employing a power-split receiver. This is followed by our analytical study in Section \ref{sec_systemanalysis}, where we demonstrate that the proposed RA pattern-based information transfer represented in the SD exhibits a beneficial immunity to RF to DC power conversion induced performance erosion. We then continue by charactering the RE trade-off of the overall system. Our simulation results are provided in Section \ref{sec_results}, while we conclude in Section \ref{sec_conclusion}.

\section{System Model}
\label{sec_systemmodel}
Our proposed GPSM aided IWIPT system relies on the availability of Channel State Information at the Transmitter (CSIT) for supporting Transmitter Pre-Coding (TPC). This scenario has been widely considered in the SWIPT literature~\cite{6489506,1211.6868v3,6571308,6211384,1307.6285v1}, since it is desirable to shift most/all signal processing demands from the less powerful Down-Link (DL) receiver that may be passive or semi-passive to the more powerful DL transmitter that may have access to the mains power.

\begin{figure}
\centering
\includegraphics[width=\linewidth]{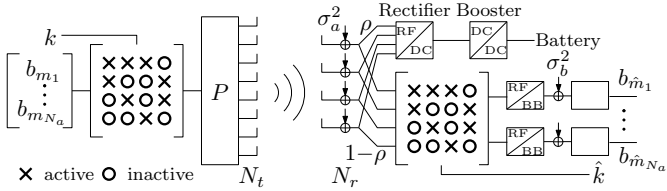}
\caption{Transceiver architecture of GPSM aided IWIPT. For example, activating $N_a = 2$ RAs out of $N_r = 4$ RAs results in a total of $|\mathcal{C}_t| = 6$ legitimate activation patterns, i.e. the patterns of $\mathcal{C}_t = \{[1,2], [1,3], [1,4], [2,3], [2,4], [3,4]\}$. This configuration delivers $k_{ant} = 2$ bits of information because $\lfloor\log_2(|6|)\rfloor=2$. Upon selecting for example $\mathcal{C} = \{\mathcal{C}_t(1),\mathcal{C}_t(2),\mathcal{C}_t(3),\mathcal{C}_t(4)\}$, we have the following mapping between the SD symbol and the activation patterns $k = 1 \mapsto \mathcal{C}(1) = [1,2]$, $k = 2 \mapsto \mathcal{C}(2) = [1,3]$, $k = 3 \mapsto \mathcal{C}(5) = [1,4]$ and $k = 4 \mapsto \mathcal{C}(6) = [2,3]$.}
\label{fig_iwipt}
\end{figure}

\subsection{Conceptual Description}
Consider a Multiple Input Multiple Output (MIMO) system equipped with $N_t$ Transmit Antennas (TA) and $N_r$ RAs, where we assume $N_t \geq N_r$. In this MIMO set-up, a maximum of $N_r = \min(N_t, N_r)$ parallel data streams may be supported by conventional waveform-based information transfer, conveying a total of $k_{eff} = N_rk_{mod}$ bits altogether, where $k_{mod} = \log_2(M)$ denotes the number of bits per symbol of a conventional $M$-ary PSK/QAM scheme and its alphabet is denoted by $\mathcal{A}$. The TPC matrix of $\pmb{P} \in \mathbb{C}^{N_t\times N_r}$ may be used for pro-actively mitigating the Inter-Channel Interference (ICI) at the transmitter by pre-processing the source signal before its transmission upon exploiting the knowledge of the CSIT. As a benefit, a low-complexity single-stream detection scheme may be used by the receiver, because the ICI is eliminated at the transmitter.

In contrast to the above-mentioned classic waveform-based multiplexing of $N_r$ data streams, in our GPSM scheme a total of $N_a < N_r$ RAs are activated so as to facilitate the simultaneous transmission of $N_a$ data streams, where the particular \textit{pattern} of the $N_a$ RAs activated conveys information in form of so-called SD symbols in addition to the information carried by the conventional $M$-ary PSK/QAM modulated symbols as seen in Fig~\ref{fig_iwipt}~\footnote{Note that in our GPSM scheme, it is not necessary to convey any additional information by a conventional $M$-ary PSK/QAM scheme, we may consider a pure spatial modulation-aided system.}. Hence, the number of bits in GPSM conveyed by a SD symbol becomes $k_{ant} = \lfloor\log_2(|\mathcal{C}_t|)\rfloor$, where the set $\mathcal{C}_t$ contains all the combinations associated with choosing $N_a$ activated RAs out of $N_r$ RAs. As a result, the total number of bits transmitted by the GPSM scheme is $k_{eff} = k_{ant}+N_ak_{mod}$. To assist our further discussions, let us define the set of selected activation patterns as $\mathcal{C} \subset \mathcal{C}_t$. We also let $\mathcal{C}(k)$ and $\mathcal{C}(k,i)$ denote the $k$th RA activation pattern and the $i$th activated RA in the $k$th activation pattern, respectively. Finally, it is plausible that the conventional MIMO scheme using purely waveform-based information transfer obeys $N_a = N_r$.

\subsection{IWIPT Transmitter}
More specifically, let $\pmb{s}^k_m$ be an \textit{explicit} representation of a so-called super-symbol $\pmb{s} \in \mathbb{C}^{N_r\times 1}$, indicating that the RA pattern $k$ is activated and $N_a$ conventional modulated symbols $\pmb{b}_m = [b_{m_1},\ldots,b_{m_{N_a}}]^T \in \mathbb{C}^{N_a\times 1}$ are transmitted, where we have $b_{m_i} \in \mathcal{A}$ and $\mathbb{E}[|b_{m_i}|^2] = 1, \forall i \in [1,N_a]$. In other words, we have the relationship 
\begin{align}
\pmb{s}^k_m &= \pmb{\Omega}_k \pmb{b}_m,
\end{align}
where $\pmb{\Omega}_k = \pmb{I}[:,\mathcal{C}(k)]$ is constituted by the specifically selected columns determined by $\mathcal{C}(k)$ of an identity matrix of $\pmb{I}_{N_r}$. Following TPC, the resultant transmit signal $\pmb{x} \in \mathbb{C}^{N_t\times 1}$ may be written as 
\begin{align}
\pmb{x} &= \sqrt{\beta_s / N_a} \pmb{P} \pmb{s}^k_m. 
\end{align}
In order to avoid any power fluctuation during the pre-processing, we introduce the scaling factor of $\beta_s$ designed for maintaining the power-constraint of $||\pmb{x}||^2 = 1$. As a natural requirement, the TPC matrix has to ensure that no energy leaks into the RA-elements of the unintended RA patterns. Hence, the classic linear Channel Inversion (CI)-based TPC~\cite{1391204,1261332} may be used, which is formulated as
\begin{align}
\pmb{P} &= \pmb{H}^H (\pmb{H} \pmb{H}^H)^{-1},
\label{eq_ci}
\end{align}
where $\pmb{H} \in \mathbb{C}^{N_r\times N_t}$ represents the MIMO channel involved. We assume furthermore that each entry of $\pmb{H}$ undergoes frequency-flat Rayleigh fading and it is uncorrelated between different super-symbol transmissions, while remains constant within the duration of a super-symbol's transmission. The power-normalisation factor of the output power after pre-processing is given by
\begin{align}
\beta_s &= N_a/\pmb{s}^H(\pmb{H} \pmb{H}^H)^{-1} \pmb{s}.
\label{eq_ci_s}
\end{align}
The Base Band (BB) signal $\pmb{x}$ is then up converted to $\pmb{x}^{RF}$ and sent through $N_t$ TAs.

\subsection{IWIPT Receiver}
\label{sec_iwipt_rx}
\subsubsection{Architecture} The power-split receiver of Fig~\ref{fig_iwipt} is employed~\cite{6489506}, where we have a power transfer tunnel and information transfer tunnel. As seen from Fig~\ref{fig_iwipt}, the received signal will be first subject to power conversion at each RA using a power-splitting ratio of $\rho \in (0,1)$ at the RF stage, where the power transfer tunnel consists of a rectifier used for converting the received RF power to DC power, which is followed by a (multi-stage) DC to DC booster. Note that from a pure WPT point of view, the RAs, the rectifier and the DC-booster are jointly known as a rectenna. After power-splitting, the remaining signal power will be used for information transfer, which relies on determining the SD symbol for RA pattern-based information transfer at the RF-stage and on detecting the conventional modulated symbols for waveform-based information transfer at the BB-stage. Note that an additional benefit of our GPSM scheme is that a reduced number of RF to BB conversion chains are required.

\subsubsection{Power Conversion}
The RF signal observed at the $N_r$ RAs may be written as 
\begin{align}
\pmb{r}^{RF} &= \pmb{H} \pmb{x}^{RF}+\pmb{w}^{RF}_a,
\label{eq_rx}
\end{align}
where $\pmb{w}^{RF}_a$ is the RF receiver's noise. The equivalent BB representation of the noise $\pmb{w}_a \in \mathbb{C}^{N_r\times 1}$ may be modelled by a circularly symmetric complex-valued Gaussian noise vector with each entry having a zero mean and a variance of $\sigma^2_a$, i.e. we have $\mathbb{E}[||\pmb{w}_a||^2] = \sigma^2_aN_r$.

At the RF-stage of Fig~\ref{fig_iwipt}, a portion $\rho$ of the received power is converted to DC, where the average energy transferred by the transmitter per unit time and gleaned from all the $N_r$ RAs is represented as
\begin{align}
Q &= \xi \mathbb{E}\left[ \rho \sum^{N_r}_{i=1}  |\tilde{r}_i^{RF}|^2 \right], 
\label{eq_eh}
\end{align}
where $\pmb{\tilde{r}}^{RF} = \pmb{H} \pmb{x}^{RF}$ is the noiseless part of $\pmb{r}^{RF}$, representing the power transferred by the transmitter, while the additive noise of $\pmb{w}^{RF}_a$ in (\ref{eq_rx}) is not transferred by the transmitter, it is imposed by the RF receiver. For the time being, designing adaptive RA specific power-splitting ratios $\rho_i, i \in [1,N_r]$ is set aside for our future investigations and we also set the RF to DC conversion efficiency to $\xi = 1$ in this paper. 

\subsubsection{RA Pattern-based Information Transfer}
After RF to DC power conversion of each RA at a ratio of $\rho$, the remaining RF signal observed at the $N_r$ RAs invoked for information transfer may be written as 
\begin{align}
\pmb{y}^{RF}_a &= \sqrt{1-\rho} \left( \pmb{H} \pmb{x}^{RF}+\pmb{w}^{RF}_a \right),
\label{eq_rx_rf}
\end{align} 
where the SD symbol may be determined by simply sorting the remaining received power accumulated by each legitimate RA pattern, which is mathematically represented as
\begin{align}
\hat{k} &= \arg \max_{\ell \in [1,|\mathcal{C}|]} \left\lbrace \sum^{N_a}_{i=1}|y_{a,\mathcal{C}(\ell,i)}^{RF}|^2 \right\rbrace. 
\label{eq_seperate1}
\end{align}
Thus, correct detection is declared when we have $\hat{k} = k$. We will later show that the factor $(1-\rho)$ does not affect the detection of SD symbols and our RA pattern-based information transfer exhibits a beneficial immunity to power conversion.

\subsubsection{Waveform-based Information Transfer}
After determining the RA pattern using the non-coherent detection process of (\ref{eq_seperate1}), we proceed with the conventional waveform-based information transfer. After RF to BB down conversion, the BB signal may be written as 
\begin{align}
\pmb{y}_b &= \sqrt{1-\rho} \left( \sqrt{\beta_s / N_a} \pmb{H} \pmb{P} \pmb{s}^k_m+\pmb{w}_a \right)+\pmb{w}_b,
\label{eq_rx_bb}
\end{align}
where $\pmb{w}_b \in \mathbb{C}^{N_r\times 1}$ is the down-conversion induced noise, which is modelled similarly to the RF receiver's noise $\pmb{w}_a$, but with a variance of $\sigma^2_b$. We also let $\sigma^2_a = \alpha \sigma^2$ and $\sigma^2_b = (1-\alpha) \sigma^2$, where $\sigma^2$ represents the total noise variance at the power-split receiver of Fig~\ref{fig_iwipt}. 

Assuming that the $v_i = \mathcal{C}(k,i)$th RA is activated, the Signal-to-Interference-plus-Noise-Ratio (SINR) $\gamma_{v_i}$ of the conventional modulated symbols after power conversion is given by
\begin{align}
\gamma_{v_i} &= \frac{||\pmb{h}_{v_i} \pmb{p}_{v_i}||^2}{[\sigma^2_a+\sigma^2_b/(1-\rho)]N_a/\beta_s},
\label{eq_sinr}
\end{align}
where $\pmb{h}_{v_i}$ is the $v_i$th row of $\pmb{H}$ representing the channel between the $v_i$th RA and the transmitter, while $\pmb{p}_{v_i}$ is the $v_i$th column of $\pmb{P}$ representing the $v_i$th TPC vector. 

As a subsequent stage after the RA pattern determination, the detection of conventional modulated symbols is formulated as
\begin{align}
\hspace*{-0.2cm}
\hat{m}_i &= \arg \min_{n_i \in [1,M]} \left\lbrace |y_{b,\hat{v}_i} - \sqrt{\frac{(1-\rho)\beta}{N_a}} \pmb{h}_{\hat{v}_i} \pmb{p}_{\hat{v}_i} b_{n_i}|^2 \right\rbrace,
\label{eq_seperate2}
\end{align}
where $\hat{v}_i = \mathcal{C}(\hat{k},i)$. Thus, correct detection is declared, when we have $\hat{m}_i = m_i, \forall i$. 

\section{System Analysis}
\label{sec_systemanalysis}
In this section, we will first show that the RA pattern-based information transfer of (\ref{eq_seperate1}) is immune to power conversion, since the Symbol Error Ratio (SER) of the SD symbols remains the same as in the absence of power transfer. This is then followed by our RE trade-off analysis of the overall GPSM scheme, which is constituted by both the SD symbol used for RA pattern-based information transfer and by the conventional modulated symbols employed for waveform-based information transfer.

\subsection{Analytical SER of RA Pattern-based Information Transfer}
We commence our discussion by directly stating:
\begin{lem}
The RA pattern-based information transfer of (\ref{eq_seperate1}) is immune to the power conversion operation of (\ref{eq_eh}) for $\rho \in (0,1)$.
\label{thm1}
\end{lem}

\begin{proof}
Firstly we note that (\ref{eq_seperate1}) relies on the principle of energy detection. Hence our analysis may be performed using the equivalent BB representation of (\ref{eq_rx_rf}), as suggested in~\cite{1447503}, which is represented as
\begin{align}
\pmb{y}_a &= \sqrt{1-\rho} \left( \sqrt{\beta_s / N_a} \pmb{H} \pmb{P} \pmb{s}^k_m+\pmb{w}_a \right).
\label{eq_rx_rf_bb}
\end{align} 
Hence (\ref{eq_seperate1}) becomes 
\begin{align}
\hat{k} &= \arg \max_{\ell \in [1,|\mathcal{C}|]} \left\lbrace \sum^{N_a}_{i=1}|y_{a,\mathcal{C}(\ell,i)}|^2 \right\rbrace. 
\label{eq_seperate11}
\end{align}
Given the values of $N_r$ as well as $N_a$ and assuming that the RA pattern $\mathcal{C}(k)$ was activated, after substituting (\ref{eq_ci}) into (\ref{eq_rx_rf_bb}), we have: 
\begin{align} 
y_{a,v_i} &= \sqrt{(1-\rho)} \left( \sqrt{\beta_s/N_a} b_{m_i}+w_{a,v_i} \right),  \hspace*{-0.4cm} &\forall v_i \in \mathcal{C}(k), \label{eq_ci_yes}\\
y_{a,u_i} &= \sqrt{(1-\rho)} w_{a,u_i},  \hspace*{-0.4cm} &\forall u_i \in \bar{\mathcal{C}}(k), \label{eq_ci_no}
\end{align} 
where $\bar{\mathcal{C}}(k)$ denotes the complementary set of the activated RA pattern $\mathcal{C}(k)$ in $\mathcal{C}$. Furthermore, upon introducing $\sigma^2_{a,0} = \sigma^2_a/2$, we have:  
\begin{align} 
|y_{a,v_i}|^2 &= \mathcal{R}(y_{a,v_i})^2 + \mathcal{I}(y_{a,v_i})^2 \nonumber \\
&\sim \mathcal{N}\left[ \sqrt{(1-\rho)\beta_s/N_a}\mathcal{R}(b_{m_i}),(1-\rho)\sigma^2_{a,0}\right] \nonumber \\ &+\mathcal{N}\left[\sqrt{(1-\rho)\beta_s/N_a}\mathcal{I}(b_{m_i}),(1-\rho)\sigma^2_{a,0}\right], \\
|y_{a,u_i}|^2 &= \mathcal{R}(w_{a,u_i})^2 + \mathcal{I}(w_{a,u_i})^2 \nonumber \\
&\sim \mathcal{N}\left[0,(1-\rho)\sigma^2_{a,0}\right]+\mathcal{N}\left[0,(1-\rho)\sigma^2_{a,0}\right], 
\end{align}
where $\mathcal{R}(\cdot)$ and $\mathcal{I}(\cdot)$ represent the real and imaginary operators, respectively. As a result, by normalisation with respect to $(1-\rho)\sigma^2_{a,0}$, we have the following observations:
\begin{align} 
|y_{a,v_i}|^2 &\sim \chi^2_2(g;\lambda_{v_i}), \qquad &\forall v_i \in \mathcal{C}(k), \label{eq_dist1}\\
|y_{a,u_i}|^2 &\sim \chi^2_2(g), \qquad &\forall u_i \in \bar{\mathcal{C}}(k), \label{eq_dist2}
\end{align}
where the non-centrality parameter is given by 
\begin{align} 
\lambda_{v_i} = \beta_s|b_{m_i}|^2/N_a\sigma^2_{a,0}.
\label{eq_non-centrality}
\end{align}
Exploiting the fact that $\mathbb{E}[|b_{m_i}|^2] = 1, \forall i$ or $|b_{m_i}|^2 = 1, \forall i$ for PSK modulation, we have
\begin{align} 
\lambda = \lambda_{v_i} = \beta_s/N_a\sigma^2_{a,0}, \qquad \forall v_i
\label{eq_non-centrality2}
\end{align}
Note that $\lambda$ is also a random variable obeying the distribution of $f_{\lambda}(\lambda)$, whose analytical expression will be given in Lemma \ref{eq_analyticallambda}. 

Recall from (\ref{eq_seperate11}) that the correct decision concerning the SD symbol occurs, when $\sum^{N_a}_{i=1}|y_{a,v_i}|^2$ is the maximum. Since the non-centrality parameter given by (\ref{eq_non-centrality}) is not a function of $\rho$, power conversion operation does not affect the quantity of $\sum^{N_a}_{i=1}|y_{a,v_i}|^2$. Hence, it becomes explicit that the power-splitting ratio $\rho$ does not affect the SER $e^s_{ant}$ of the SD symbol of our GPSM scheme for RA pattern-based information transfer expressed in Lemma~\ref{eq_ser}.
\end{proof}

\begin{lem}(Proof see Appendix C of~\cite{Zhang}) The distribution $f_\lambda(\lambda)$ of the non-centrality parameter $\lambda$ is given by:
\begin{align}
f_\lambda(\lambda) &= \frac{N_a^{N_t-N_r+1}\sigma^2_a/2}{(N_t-N_r)!} e^{-\lambda N_a\sigma^2_a/2}(\lambda \sigma^2_a/2)^{N_t-N_r}. 
\label{eq_lambda}
\end{align}
\label{eq_analyticallambda}
\end{lem}

\begin{lem}(Proof see Theorem III.1 of~\cite{Zhang}) The analytical SER $e^s_{ant}$ of the SD symbol of our GPSM scheme relying on CI TPC may be formulated as:
\begin{align}
e^s_{ant} \approx 1-\int^{\infty}_{0} \left\lbrace \int^{\infty}_{0} [F_{\chi^2_2}(g)]^{N_r-N_a} f_{\chi^2_2}(g;\lambda)dg \right\rbrace^{N_a} f_\lambda(\lambda) d\lambda,
\label{eq_analyticalssk2}
\end{align}
where $F_{\chi^2_2}(g)$ represents the Cumulative Distribution Function (CDF) of a chi-square distribution having two degrees of freedom, while $f_{\chi^2_2}(g;\lambda)$ represents the Probability Distribution Function (PDF) of a non-central chi-square distribution having two degrees of freedom and non-centrality given by $\lambda$ with its PDF of $f_\lambda(\lambda)$. 
\label{eq_ser}
\end{lem}

\subsection{Analytical SER of Waveform-based Information Transfer}
We first introduce the SER $e^s_{mod}$ of the conventional modulated symbols $b_{m_i} \in \mathcal{A}$ for waveform-based information transfer in the \textit{absence} of SD symbol errors. For a specific activated RA pattern in $\mathcal{C}(k)$, the SINR of (\ref{eq_sinr}) becomes 
\begin{align} 
\gamma &= \gamma_{v_i} = \beta_s/N_a[\sigma^2_a+\sigma^2_b/(1-\rho)], \qquad \forall v_i
\label{eq_snr}
\end{align}
and for the remaining deactivated RAs in $\bar{\mathcal{C}}(k)$ we have only random noise contributions of zero mean and of a variance of $(1-\rho)\sigma^2_a+\sigma^2_b$. Thus the SER $e^s_{mod}$ may be upper bounded by~\cite{JohnProakis2008}: 
\begin{align}
e^s_{mod} &\leq N_{min} \int^{\infty}_{0} \mathcal{Q} \left( d_{min} \sqrt{\gamma/2} \right) f_\gamma(\gamma)d\gamma,
\label{eq_analyticalmod}
\end{align}
where $f_\gamma(\gamma)$ is a scaled version of $f_\lambda(\lambda)$ given in Lemma \ref{eq_analyticallambda}, i.e. we have $f_\gamma(\gamma) = \kappa f_\lambda(\kappa\gamma)$, where
\begin{align}
\kappa = 2[\sigma^2_a+\sigma^2_b/(1-\rho)]/\sigma^2_a.
\end{align}
Moreover, $d_{min}$ is the minimum Euclidean distance in the conventional modulated symbol constellation, $N_{min}$ is the average number of the nearest neighbours separated by $d_{min}$ in the constellation and $\mathcal{Q}(\cdot)$ denotes the Gaussian $\mathcal{Q}$-function. 

We then introduce $\tilde{e}^s_{mod}$ for representing the SER of the conventional modulated symbols in the \textit{presence} of SD symbol errors due to the detection of (\ref{eq_seperate1}), which is formally given by

\begin{lem}(Proof see Appendix A of~\cite{Zhang}) Given the $k$th activated RA patten, the SER of the conventional modulated symbols for waveform-based information transfer in the \textit{presence} of SD symbol errors can be calculated as:
\begin{align}
\tilde{e}^s_{mod} &= (1-e^s_{ant})e^s_{mod}+\frac{e^s_{ant}}{(2^{k_{ant}}-1)}\sum_{\ell \neq k} \underbrace{\frac{N_{c}e^s_{mod}+N_d e^s_o}{N_a}}_{B},
\label{eq_convertion}
\end{align}
where $N_c$ and $N_d = (N_a-N_c)$ represent the number of RAs that are common and different between $\mathcal{C}(\ell)$ and $\mathcal{C}(k)$, respectively. Mathematically we have:
\begin{align}
N_c &= \sum^{N_a}_{i=1} \mathbb{I} \left[ \mathcal{C}(\ell,i) \in \mathcal{C}(k) \right].
\end{align}
Moreover, $e^s_o$ is defined as the modulation-type-dependent SER. 
\label{lem_convertion}
\end{lem}
Finally, when CI based TPC is introduced, only random noise may be received by the $N_d$ RAs in $\mathcal{C}(\ell)$, thus $e^s_o$ simply represents the SER of the conventional modulated symbols as a result of a random guess. For example, for QPSK, we have $e^s_o = 3/4$. Note that, when the SER of the SD symbols $e^s_{ant}$ is small, we have $\tilde{e}^s_{mod} \approx e^s_{mod}$. 

\subsection{Rate Energy Trade-off}
\subsubsection{DCMC Capacity Expression}
The RE trade-off may be formulated using the power conversion of (\ref{eq_eh}) and comparing it directly against the Discrete-input Continuous-output Memoryless Channel (DCMC) capacity of our GPSM scheme. This is because in our GPSM scheme the SD symbols convey integer values constituted by the RA pattern index, which does not obey the shaping requirements of Gaussian signalling~\cite{ThomasM.Cover2006}. Hence, for simplicity's sake, we discuss the DCMC capacity of our GPSM scheme in the context of discrete-input signalling for both the SD symbol and for the conventional modulated symbols mapped to it. 

The Mutual Information per Bit (MIB) $I(z;\hat{z})$ of our GPSM scheme measured between the input bits $z \in [0,1]$ and the corresponding demodulated output bits $\hat{z} \in [0,1]$ is given by:
\begin{align}
\hspace*{-0.2cm}
I(z;\hat{z}) &= -\sum_{z} P_z \log P_z - \sum_{\hat{z}} P_{\hat{z}} \left[ \sum_{z} P_{z|\hat{z}} \log P_{z|\hat{z}} \right] \nonumber \\
&= 1 + e^b_{eff} \log e^b_{eff} + (1-e^b_{eff}) \log(1-e^b_{eff}),
\label{eq_mi}
\end{align}
where $P_z$ is the Probability Mass Function (PMF) of $z$, where we adopt the common assumption of equi-probability bits of $P_{z=0} = P_{z=1} = 1/2$. On the other hand, the conditional entropy, representing the average uncertainty about $z$ after observing $\hat{z}$, is a function of the average Bit Error Ratio (BER) of our GPSM scheme, denoted as $e^b_{eff}$. As a result, the so-called DCMC capacity becomes 
\begin{align}
R &= k_{eff}I(z;\hat{z}).
\label{eq_dcmc}
\end{align}

\subsubsection{Average BER Expression}
Let us now discuss the expression of $e^b_{eff}$ for the sake of calculating (\ref{eq_mi}) and (\ref{eq_dcmc}), which may be derived from its accurate SER expression of both the SD symbol for the RA pattern-based information transfer given in (\ref{eq_analyticalssk2}) and for the conventional modulated symbols for waveform-based information transfer given in (\ref{eq_convertion}).

Let $e^b_{ant}$ and $\tilde{e}^b_{mod}$ represent the BER of the SD symbol for RA pattern-based information transfer and of the conventional modulated symbols for waveform-based information transfer. Hence, the accurate expression of the average BER of our GPSM scheme may be written as:
\begin{align}
e^b_{eff} &= (k_{ant}e^b_{ant}+N_a k_{mod} \tilde{e}^b_{mod})/k_{eff} \nonumber \\
&= (\delta_{ant} e^s_{ant}+N_a \tilde{e}^s_{mod})/k_{eff}, 
\label{eq_ber}
\end{align}
where the second equation of (\ref{eq_ber}) obeys from the relationship of 
\begin{align}
\tilde{e}^b_{mod} &= \tilde{e}^s_{mod}/k_{mod}, \label{eq_ber_mod} \\
e^b_{ant} &= \delta_{k_{ant}} e^s_{ant}/k_{ant}, \label{eq_ber_ant}
\end{align}
where we formulate Lemma \ref{lem_correction} for the expression of $\delta_{k_{ant}}$ acting as a correction factor in (\ref{eq_ber}).

\begin{lem}(Proof see Appendix B of~\cite{Zhang}) The generic expression of the correction factor $\delta_{k_{ant}}$ for $k_{ant}$ bits of information is given by:
\begin{align}
\delta_{k_{ant}} &= \delta_{k_{ant}-1} + \frac{2^{k_{ant}-1}-\delta_{k_{ant}-1}}{2^{k_{ant}}-1},
\label{eq_correction}
\end{align}
where given $\delta_{0} = 0$, we can recursively determine $\delta_{k_{ant}}$.
\label{lem_correction}
\end{lem}

\subsection{Further Discussions}
\subsubsection{Robustness}
The above RE trade-off is based on the assumption of having perfect CSIT. Let us now consider the system's robustness to both CSIT-errors and antenna correlations, where the expression $e^b_{eff}$ of our GPSM scheme will be determined empirically for calculating (\ref{eq_mi}) and (\ref{eq_dcmc}).  

Like in all TPC schemes, an important related aspect is its resilience to CSIT inaccuracies. In this paper, we let $\pmb{H} = \pmb{H}_t+\pmb{H}_e$, where $\pmb{H}_t$ represents the matrix hosting the average CSI, with each entry obeying the complex Gaussian distribution of $h_t \sim \mathcal{CN}(0,\sigma^2_t)$ and $\pmb{H}_e$ is the instantaneous CSI error matrix obeying the complex Gaussian distribution of $h_e \sim \mathcal{CN}(0,\sigma^2_e)$, where we have $\sigma^2_t+\sigma^2_e = 1$. As a result, only $\pmb{H}_t$ is available at the transmitter for pre-processing. There is a plethora of CSIT inaccuracy counter-measures conceived for general transmitter pre-coding schemes in the literature~\cite{1237415,4567684}, but they are beyond the scope of our discussions in this paper.

Another practical phenomenon is the presence of correlation amongst the signals at the transmit and receive antennas. The correlated MIMO channel is modelled by the widely-used Kronecker model~\cite{1021913}, which is formulated as $\pmb{H} = (\pmb{R}^{1/2}_t)\pmb{G}(\pmb{R}^{1/2}_r)^T$, with $\pmb{G}$ representing the original MIMO channel imposing no correlation, while $\pmb{R}_t$ and $\pmb{R}_r$ represents the correlations at the transmitter and receiver side, respectively, with the correlation entries given by $R_t(i,j) = \rho_t^{|i-j|}$ and $R_r(i,j) = \rho_r^{|i-j|}$. 

\subsubsection{Asymptotic Insights}
It is also interesting to investigate the impact of power conversion on both the MIB as well as on the DCMC capacity of our GPSM scheme and to contrast it to that of the conventional CI pre-coding based MIMO scheme for large-scale MIMO settings of $N_t \rightarrow \infty$, $N_r \rightarrow \infty$ and for a fixed system load ratio of $N_r/N_t$. 

In order to provide these asymptotic insights for these large-scale MIMO settings, we have to replace the constraint of $||\pmb{x}||^2 = 1$ with its relaxed version of $\mathbb{E}[||\pmb{x}||^2] = 1$. Hence instead of (\ref{eq_ci_s}), we have the relaxed normalisation factor of
\begin{align}
\beta_l &= N_r/\mathrm{Tr}[(\pmb{H} \pmb{H}^H)^{-1}].
\label{eq_ci_l}
\end{align}
As shown in Corollary 4 of~\cite{6172680} and also in Chapter 14 of~\cite{Couillet:2011:RMM:2161679}, under asymptotic settings, (\ref{eq_ci_l}) tends to
\begin{align}
\beta_l \longmapsto (N_t/N_r-1).
\label{eq_deterministic}
\end{align}
By exploiting (\ref{eq_deterministic}) and by replacing (\ref{eq_ci_s}) with (\ref{eq_ci_l}) in the expressions of (\ref{eq_non-centrality2}) and (\ref{eq_snr}), we have the following deterministic quantity 
\begin{align}
\lambda &\longmapsto \lambda_d = (N_t/N_r-1)/N_a\sigma^2_{a,0},  \\
\gamma &\longmapsto \gamma_d = (N_t/N_r-1)/N_a[\sigma^2_a+\sigma^2_b/(1-\rho)]. 
\end{align} 
As a result, we have
\begin{align}
e^s_{ant,d} &= 1-\left\lbrace \int^{\infty}_{0} [F_{\chi^2_2}(g)]^{N_r-N_a} f_{\chi^2_2}(g;\lambda_d)dg \right\rbrace^{N_a}, \label{a1} \\
e^s_{mod,d} &= N_{min} \mathcal{Q} \left( d_{min} \sqrt{\gamma_d/2} \right). \label{a2}
\end{align} 
Noting that when $e^s_{ant,d}$ is small, we arrive at the approximation $\tilde{e}^s_{mod,d} \approx e^s_{mod,d}$ as suggested by (\ref{eq_convertion}), thus we may obtain $e^b_{eff,d}$ by replacing $e^s_{ant}$ and $\tilde{e}^s_{mod}$ in (\ref{eq_ber}) with their deterministic expressions in (\ref{a1}) and (\ref{a2}). Finally, by inserting $e^b_{eff,d}$ in (\ref{eq_mi}) and (\ref{eq_dcmc}) instead of $e^b_{eff}$, we arrive at the MIB and DCMC capacity for both our GPSM scheme as well as for the conventional CI pre-coding based MIMO scheme ($N_a = N_r$) under the asymptotic settings of $N_t \rightarrow \infty$ and $N_r \rightarrow \infty$.

\section{Numerical Results}
\label{sec_results}
Let us now consider our numerical results for characterizing the GPSM aided IWIPT for both RA pattern-based information transfer and for waveform-based information transfer relying on the CI-based precoding of (\ref{eq_ci}) under $\{N_t, N_r\} = \{16,8\}$. Since our focus is on RA pattern-based information transfer, simple QPSK was used for the conventional data modulation component of our waveform-based information transfer. Furthermore, we investigate the scenarios correspond to $N_a = [1,\ldots,6]$ in our GPSM scheme since the setting of $N_a = 7$ results into higher complexity than the conventional CI pre-coding based MIMO bench-maker as discussed in~\cite{6644231}. Maximum achievable throughput under MIMO settings of $\{N_t, N_r\} = \{16,8\}$ and employing QPSK modulation is shown in Table~\ref{table1}, which may be used to compare our DCMC capacity related investigations. Finally, the Signal-to-Noise-Ratio (SNR) per bit is defined as $SNR_b = 1/\sigma^2(k_{eff}/N_a)$. 

\begin{table}[t]
\caption{Maximum achievable throughput $k_{eff}$ of GPSM and of the conventional MIMO corresponding to $N_a = N_r$ (in $\bullet$) under settings of $\{N_t, N_r\} = \{16,8\}$ and employing QPSK modulation $M=2$, where $k_{eff} = k_{ant}+N_ak_{mod}$ for GPSM and $k_{eff} = N_rk_{mod}$ for conventional MIMO.}
\centering
\begin{tabular}{rccccc}
\hline 
$N_t$ & $N_r$ & $N_a$ & $M$ & $k_{eff}$ \\ \hline \hline
16 & 8 & 1 & 2 & 5 = 3+1$\times$2  \\
16 & 8 & 2 & 2 & 8 = 4+2$\times$2  \\
16 & 8 & 3 & 2 & 11 = 5+3$\times$2 \\
16 & 8 & 4 & 2 & 14 = 6+4$\times$2 \\
16 & 8 & 5 & 2 & 15 = 5+5$\times$2 \\
16 & 8 & 6 & 2 & 16 = 4+6$\times$2  \\ 
$\bullet$16 & 8 & 8 & 2 & 16 = 0+8$\times$2 \\
\hline	
\end{tabular}
\label{table1}
\end{table}

\subsection{Investigation on Error Probability}
Fig \ref{fig_ci_i} shows the spatial SER (upper) as well as the GPSM scheme's overall BER (lower) under the power-normalisation factor of (\ref{eq_ci_s}) for $\{N_t, N_r\} = \{16,8\}$. Without loss of generality, we set the power-splitting ratio to $\rho = 0$ and the noise ratio to $\alpha = 1$. 

Observe in the upper subplot of Fig \ref{fig_ci_i} that our empirical simulation based SER results recorded for the SD symbol of the RA pattern-based information transfer accurately match the analytical results given in (\ref{eq_analyticalssk2}) and form tight upper bounds. Similarly, when the GPSM scheme's overall BER is considered in the lower subplot of both figures, our empirical results accurately match the analytical results given in (\ref{eq_ber}). 

\begin{figure}
\centering
\includegraphics[width=\linewidth]{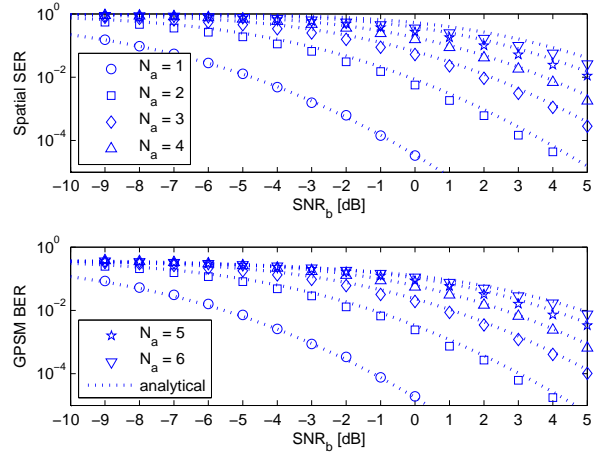}
\caption{Spatial SER of (\ref{eq_analyticalssk2}) (upper) and GPSM overall BER of (\ref{eq_ber}) (lower) using CI based TPC and a power-normalisation factor of (\ref{eq_ci_s}) under $\{N_t, N_r\} = \{16,8\}$, where we set the power-splitting ratio to $\rho = 0$ and the noise ratio to $\alpha = 1$.}
\label{fig_ci_i}
\end{figure}

Fig \ref{fig_ber} shows the GPSM scheme's overall BER under the power-normalisation factor of (\ref{eq_ci_s}) for $\{N_t, N_r\} = \{16,8\}$, where we set the power-splitting ratio to $\rho = 0.5$ and the noise ratio to $\alpha = 0.4$. As suggested by (\ref{eq_snr}), the total noise at the BB for waveform-based information transfer now becomes 1.6 times that with no WPT. Hence, we may expect to see a BER performance degradation of $10\log_{10}(1.6) \approx 2$dB. This is clearly demonstrated in Fig \ref{fig_ber} at BER of $10^{-5}$, which implies that the RA pattern-based information transfer in our GPSM scheme is immune to the RF to DC power conversion induced performance erosion as otherwise the performance degradation would be more than 2dB.

\begin{figure}
\centering
\includegraphics[width=\linewidth]{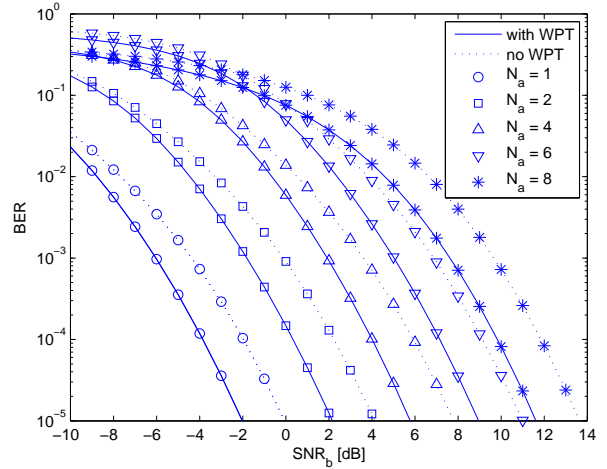}
\caption{BER of conventional CI pre-coding based MIMO bench-maker as well as GPSM using CI based TPC and a power-normalisation factor of (\ref{eq_ci_s}) under $\{N_t, N_r\} = \{16,8\}$, where we set the power-splitting ratio to $\rho = 0.5$ and the noise ratio to $\alpha = 0.4$.}
\label{fig_ber}
\end{figure}

\subsection{Investigation on MIB}
\subsubsection{Immunity to Power Conversion}
Fig \ref{fig_mie_ant} investigates the MIB versus energy trade-off of purely RA pattern-based information transfer for $\{N_t, N_r\} = \{16,8\}$ and $SNR_b = 0$dB, where we set $\alpha = 0.4$ in the left subplot and $\alpha = 0.6$ in the right subplot. The results of MIB were evaluated from (\ref{eq_mi}).

Observe from both subplots of Fig \ref{fig_mie_ant} that for all values of $N_a$, apart from setting $\rho = 1$, where the receiver acts purely as a WPT and the normalised converted energy is unity, the remaining settings of the power-splitting ratio $\rho$ do not affect the MIB of the SD symbols in our GPSM scheme, as suggested by the horizontal lines. More explicitly, for all values of $N_a$ seen in Fig \ref{fig_mie_ant}, the MIB of SD symbols vs the power conversion ratio of $\rho \in [0,0.1,\ldots,0.9]$ remains identical to that under $\rho = 0$. This implies that the RA pattern-based information transfer relying on SD symbols operates as if no WPT was provided, which justifies our Lemma~\ref{thm1}. Furthermore, by comparing the left and right subplot of Fig \ref{fig_mie_ant}, it can be seen that the MIB of SD symbols becomes lower, when an increased RF receiver noise is encountered in conjunction with $\alpha = 0.6$. Note that our RA pattern-based information transfer is immune to power conversion, but naturally not to the RF receiver noise. 

\begin{figure}
\centering
\includegraphics[width=\linewidth]{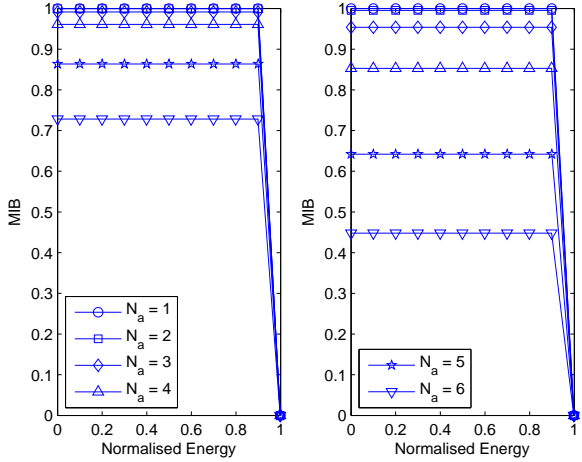}
\caption{MIB versus energy trade-off curve of RA pattern-based information transfer under $\{N_t, N_r\} = \{16,8\}$ and $SNR_b = 0$dB with $\alpha = 0.4$ (left) and $\alpha = 0.6$ (right). The MIB results were evaluated from (\ref{eq_mi}).}
\label{fig_mie_ant}
\end{figure}

\subsubsection{MIB versus Energy Trade-off}
Fig \ref{fig_mie} investigates the MIB versus energy trade-off of our GPSM aided IWIPT system supporting both RA pattern-based information transfer and waveform-based information transfer for $\{N_t, N_r\} = \{16,8\}$ and $SNR_b = 0$dB, where we set $\alpha = 0.4$ in the left subplot and $\alpha = 0.6$ in the right subplot. The results of MIB were evaluated from (\ref{eq_mi}).

Observe from both subplots of Fig \ref{fig_mie} that the near-constant MIB curves of our GPSM schemes shown in Fig \ref{fig_mie_ant} can no longer be maintained upon increasing the number of activated RAs, although they exhibit a less eroded MIB than the conventional CI pre-coding based MIMO scheme associated with $N_a = 8$ relying on pure waveform-based information transfer. This is because when waveform-based information transfer is also invoked, the resultant MIB of our GPSM scheme becomes more sensitive to the power conversion ratio than the purely RA pattern-based information transfer. This is evidenced in Fig \ref{fig_mie}, where the higher the value of $N_a$, the more conventional modulated symbols are conveyed. When comparing the left subplot to the right subplot, the MIB of our GPSM scheme appears more sensitive to the value of $\alpha$. In other words, when the RF receiver's noise is more dominant ($\alpha = 0.6$), the MIB of our GPSM schemes is reduced. On the other hand, the conventional CI pre-coding based MIMO scheme is less sensitive to the value of $\alpha$.

\begin{figure}
\centering
\includegraphics[width=\linewidth]{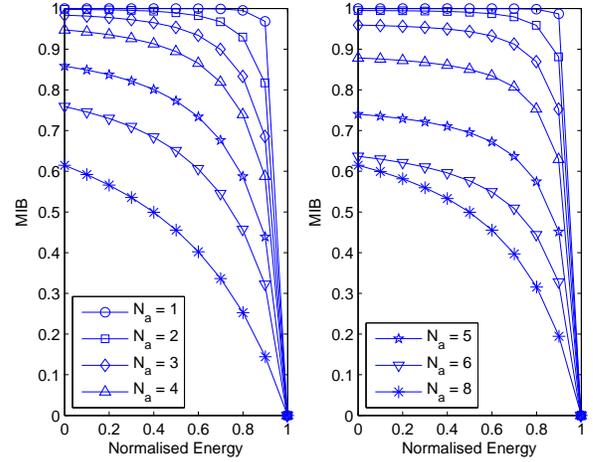}
\caption{MIB versus energy trade-off curve of our GPSM aided IWIPT system relying on both RA pattern-based information transfer and on waveform-based information transfer under $\{N_t, N_r\} = \{16,8\}$ and $SNR_b = 0$dB using $\alpha = 0.4$ (left) and $\alpha = 0.6$ (right). The MIB results were evaluated from (\ref{eq_mi}).}
\label{fig_mie}
\end{figure}

\subsection{Investigation on DCMC Capacity}
\subsubsection{DCMC Capacity versus Energy Trade-off}
Fig \ref{fig_re} investigates the DCMC capacity versus energy trade-off of our GPSM aided IWIPT system, where the rest of the settings were kept the same as in Fig \ref{fig_mie}. We can see that although Fig \ref{fig_mie} suggests that the MIB of our GPSM schemes is higher than that of the conventional CI pre-coding based MIMO scheme for all power-splitting ratios for all values of $N_a$, when the DCMC capacity is considered in Fig \ref{fig_re}, the GPSM curves associated with $N_a = 1$ and $N_a = 2$ exhibit intersections with the curve of the conventional CI pre-coding based MIMO scheme. By contrast, the GPSM schemes using $N_a \leq 3$ offer a DCMC capacity, which is higher than that exhibited by the conventional CI pre-coding based MIMO scheme. Finally, all of our GPSM schemes exhibit a higher degree of immunity to power conversion than the conventional CI pre-coding based MIMO scheme.

\begin{figure}
\centering
\includegraphics[width=\linewidth]{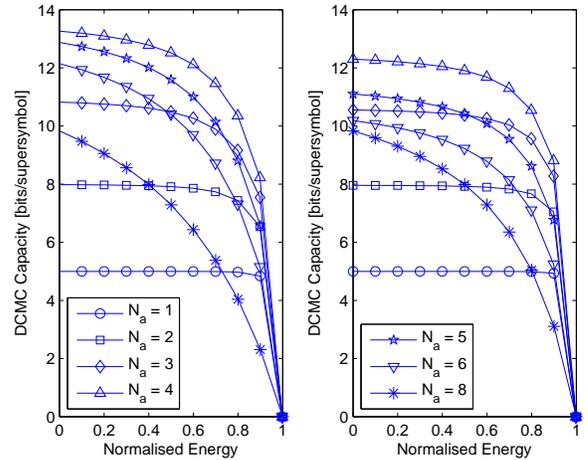}
\caption{DCMC capacity versus energy trade-off curve of our GPSM aided IWIPT system relying on both RA pattern-based information transfer and waveform-based information transfer under $\{N_t, N_r\} = \{16,8\}$ and $SNR_b = 0$dB for $\alpha = 0.4$ (left) and $\alpha = 0.6$ (right). The DCMC capacity results were evaluated from (\ref{eq_dcmc}) and these DCMC capacity results may be compared to the maximum achievable throughput values listed in Table~\ref{table1}.}
\label{fig_re}
\end{figure}

\subsubsection{DCMC Capacity versus $SNR_b$}
It is thus interesting to observe in detail the behaviour of our GPSM scheme associated with $N_a = 1$ and $N_a = 2$. Fig \ref{fig_dcmc} compares the DCMC capacity versus SNR per bit curve of the conventional CI pre-coding based MIMO scheme relying on $N_a = 8$ to our GPSM schemes associated with $N_a = 1$ and $N_a = 2$ supporting both RA pattern-based information transfer and waveform-based information transfer under $\{N_t, N_r\} = \{16,8\}$ using $\alpha = 0.4$ (left) and $\alpha = 0.6$ (right). Observe for each group of curves (circle $N_a = 1$, square $N_a = 2$ and no legend $N_a = 8$) in both subplots of Fig \ref{fig_dcmc} that the curves spanning from top to bottom correspond to power-splitting ratios of $\rho = 0.2,0.4,0.6,0.8$. 

In can also be seen in Fig \ref{fig_dcmc} that our GPSM schemes using $N_a = 1$ and $N_a = 2$ offer a higher DCMC capacity than that facilitated by the conventional CI pre-coding based MIMO scheme of $N_a = 8$ in the low $SNR_b$ range. Furthermore, observe in both subplots of Fig \ref{fig_dcmc} that the DCMC capacity curves of our GPSM schemes using $N_a = 1$ and $N_a = 2$ are only slightly affected upon increasing the power-splitting ratio $\rho$, while we observe a dramatic difference in the DCMC capacity curves of the conventional CI pre-coding based MIMO scheme of $N_a = 8$. Thus, our GPSM schemes using $N_a = 1$ and $N_a = 2$ is significantly more immune to the performance erosion of power conversion for a wide range of $SNR_b$ values.

\begin{figure}
\centering
\includegraphics[width=\linewidth]{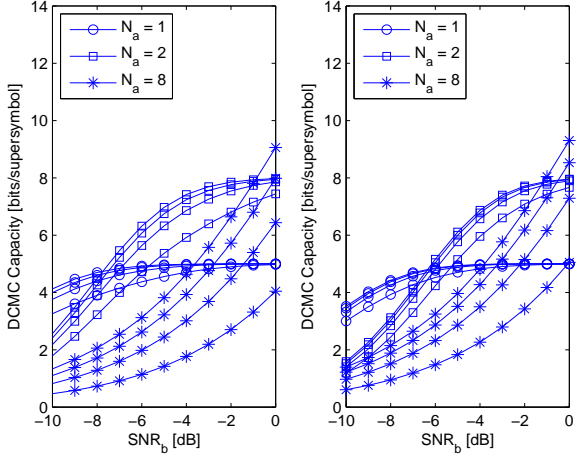}
\caption{DCMC capacity versus SNR per bit curve of our GPSM aided IWIPT system relying on both RA pattern-based information transfer and waveform-based information transfer under $\{N_t, N_r\} = \{16,8\}$ with $\alpha = 0.4$ (left) and $\alpha = 0.6$ (right). For each group of curves (circle, square and no legend), the curves spanning from top to bottom correspond to $\rho = 0.2,0.4,0.6,0.8$ in both plots. The DCMC capacity results were evaluated from (\ref{eq_dcmc}) and these DCMC capacity results may be compared to the maximum achievable throughput values listed in Table~\ref{table1}.}
\label{fig_dcmc}
\end{figure}

\subsection{Further Investigations}
\subsubsection{Robustness Investigation}
Fig \ref{fig_impairment} investigates the DCMC capacity versus energy trade-off curves of our GPSM aided IWIPT system under $\{N_t, N_r\} = \{16,8\}$ and $SNR_b = 0$dB using $\alpha = 0.4$ and suffering from imperfect CSIT associated with a variance of $\sigma_e = 0.2$ (left) or when experiencing an antenna correlation associated with $\rho_t = \rho_r = 0.4$ (right). For reasons of space-economy and to avoid crowded figures, only the results of $N_a = \{2,4,6,8\}$ are shown here. Observe from both subplots of Fig \ref{fig_impairment} that as expected, a DCMC capacity degradation is imposed both on our GPSM schemes and on the conventional CI pre-coding based MIMO scheme of $N_a = 8$. Nonetheless, our previous conclusions still hold, namely that our GPSM scheme is more immune to any power conversion induced performance erosion, which is an explicit benefit of our RA pattern-based information transfer, as evidenced in Fig \ref{fig_impairment}. 

\begin{figure}
\centering
\includegraphics[width=\linewidth]{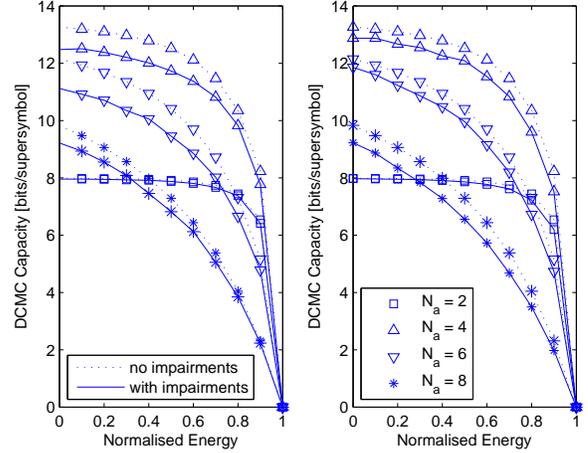}
\caption{DCMC capacity versus energy trade-off curve for our GPSM aided IWIPT system relying on both RA pattern-based information transfer and waveform-based information transfer under $\{N_t, N_r\} = \{16,8\}$ and $SNR_b = 0$dB with $\alpha = 0.4$ suffering from imperfect CSIT having a variance of $\sigma_e = 0.2$ (left) and antenna correlation associated with $\rho_t = \rho_r = 0.4$ (right). The DCMC capacity results were evaluated from (\ref{eq_dcmc}) and these DCMC capacity results may be compared to the maximum achievable throughput values listed in Table~\ref{table1}.}
\label{fig_impairment}
\end{figure}

\subsubsection{Asymptotic Investigation}
Fig \ref{fig_massive} characterizes the impact of power conversion both on the MIB (left) and on the DCMC capacity (right) for our GPSM scheme associated with $N_a = 1$ and $N_a = 2$ as well as for the conventional CI pre-coding based MIMO scheme using $N_a = N_r$ under asymptotic settings of $N_t = 2048$ and $N_r \in [64,2046]$. Furthermore, $SNR_b = 0$dB and $\alpha = 0.4$ were set in Fig \ref{fig_massive}. For each group of curves ($N_a = 1$, $N_a = 2$ and $N_a = N_r$) seen in both subplots of Fig \ref{fig_massive}, the curves spanning from top to bottom correspond to $\rho = 0.2,0.4,0.6,0.8$.

It can be seen from the left subplot of Fig \ref{fig_massive} that for all values of $N_r/N_t$, the MIB of both of our GPSM schemes is higher than that of the conventional CI pre-coding based MIMO scheme. Furthermore, the common MIB trend for all schemes is monotonically decreasing upon increasing $N_r/N_t$. However, we observe in the right subplot of Fig \ref{fig_massive} that upon increasing $N_r/N_t$, the DCMC capacity of the conventional CI pre-coding based MIMO scheme is monotonically decreasing, while the best value of $N_r/N_t$ may be deemed to be the one that attains the highest DCMC capacity in our GPSM schemes. Specifically, the DCMC capacity of the conventional CI pre-coding based MIMO scheme of $N_a = N_r$ is seen to be higher than that of our GPSM schemes for small values of $N_r/N_t$. Furthermore, our GPSM arrangement using $N_r = 2$ provides a higher DCMC capacity than the GPSM scheme relying on $N_r = 1$, when $N_r/N_t$ is lower than about 0.5, but this trend is reversed otherwise. 

Finally, a substantial DCMC capacity penalty is observed for the conventional CI pre-coding based MIMO regime of $N_a = N_r$ upon increasing the power-splitting ratio according to $\rho = 0.2,0.4,0.6,0.8$, while this trend is less obvious in both of our GPSM schemes. For example, both A and B in the right subplot of Fig \ref{fig_massive} show that an obvious DCMC capacity penalty may only be imposed for $\rho = 0.8$. This again implies that our GPSM aided IWIPT system is less sensitive to the power conversion-induced performance erosion in the asymptotic regime.

\begin{figure}
\centering
\includegraphics[width=\linewidth]{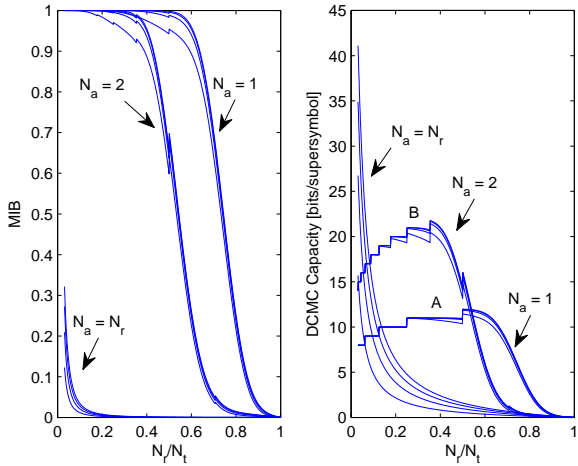}
\caption{The impact of power conversion on the MIB (left) and DCMC capacity (right) of both our GPSM scheme using $N_a = 1$ and $N_a = 2$ as well as the conventional CI pre-coding based MIMO scheme of $N_a = N_r$ under the asymptotic settings of $N_t = 2048$ and $N_r \in [64,2046]$. Furthermore, we set $SNR_b = 0$dB and $\alpha = 0.4$. For each group of curves ($N_a = 1$, $N_a = 2$ and $N_a = N_r$), the curves spanning from top to bottom correspond to $\rho = 0.2,0.4,0.6,0.8$ in both subplots. The MIB and DCMC capacity results were evaluated from (\ref{eq_mi}) and (\ref{eq_dcmc}).}
\label{fig_massive}
\end{figure}

\section{Conclusions}
\label{sec_conclusion}
A novel GPSM aided IWIPT system was proposed, where both RA pattern-based information transfer as well as waveform-based information transfer using conventional PSK/QAM modulation were employed. The analytical SER of the SD symbols and the BER of the overall GPSM scheme were derived. Furthermore, the MIB vs energy trade-off and the DCMC capacity vs energy trade-offs were discussed. Both the analytical and simulation results demonstrate that the RA pattern-based information transfer represented in the SD exhibits a beneficial immunity to power conversion and improves the overall RE trade-off, when the additional waveform-based QPSK information transfer is also taken into account. This conclusion holds true also both in the presence of inaccurate CSIT and of antenna correlations, even for large-scale MIMO settings. Hence, our proposed GPSM aided IWIPT system may be viewed as an instantiation of joint information and energy transfer in the spirit of Varshney's seminal concept~\nocite{6283481}. 

\section{List of Publications}
~\nocite{6283481,rong1,rong2,7542205,7506307,7437435,7096279,7010910,7180515,6877673,6644231,6685605,6470760,5640682,5692128,5522468,5378540,5641646,5337995,4663889,4804718,4813278,5161292,4451790,4460776,rong3,rong4,rong5,Wang2016,7802594,7542524,7465687,Wang:16,7533478,7552518,7247754,7217842,7468514,7482661,7061488,7217841,7056535,6933944,7124415,7008542,7018092,6670119,6670094,6476610,6363620,Li:13,Zhou:13,Zhou:12,6145718,6226483,6036194,4682680,5336786,Li,7249136,6399185,6214337,6213936,6093299,6093303,5779142,5779355,5779343,5594187,5594577,5594090,5594215,5502467,5493939,5208101,5073827,5073801,4698519,4657193,4533934,4525970,4489118,4224386}

%
\bibliographystyle{IEEEtran}
\bibliography{ref,rong_pub}

\end{document}